\documentclass{llncs}
\usepackage{amsmath,amssymb}
\usepackage{graphicx,hyperref}
\usepackage{cite}
\usepackage{microtype}

\newif\iffull\fulltrue

\let\doendproof\endproof
\renewcommand\endproof{~\hfill\qed\doendproof}

\newcommand{\crn}{\operatorname{cr}}
\newcommand{\paircr}{\operatorname{cr}_{\operatorname{pair}}}

\title{The Effect of Planarization on Width}
\author{David Eppstein\thanks{Supported in part by the National Science Foundation under Grants CCF-1228639, CCF-1618301, and CCF-1616248. The author is grateful to Glencora Borradaile, Erin Chambers, and Amir Nayyeri for discussions that helped clarify the distinctions between some of the width parameters considered here.}}
\institute{Department of Computer Science, University of California, Irvine}

\begin{document}
\maketitle

\begin{abstract}
We study the effects of planarization (the construction of a planar diagram $D$ from a non-planar graph $G$ by replacing each crossing by a new vertex) on graph width parameters.
We show that for treewidth, pathwidth, branchwidth, clique-width, and tree-depth there exists a family of $n$-vertex graphs with bounded parameter value, all of whose planarizations have parameter value $\Omega(n)$.
However, for bandwidth, cutwidth, and carving width, every graph with bounded parameter value has a planarization of linear size whose parameter value remains bounded. The same is true for the treewidth, pathwidth, and branchwidth of graphs of bounded degree.
\end{abstract}

\section{Introduction}

Planarization is a graph transformation, standard in graph drawing, in which a given graph $G$ is
drawn in the plane with simple crossings of pairs of edges, and then each crossing of two edges in the drawing is replaced by a new dummy vertex, subdividing the two edges~\cite{Lei-FOCS-81,DiBDidMar-GD-01,BucChiGut-HGD-14,GarSha-JSL-71}. This should be distinguished from a different problem, also called planarization, in which we try to find a large planar subgraph of a nonplanar graph~\cite{JayThuSwa-TCAD-8,Cim-JIOS-97,ChuMakSid-SODA-11,BorEppZhu-JGAA-15}. A given graph $G$ may have many different planarizations, with different properties. Although the size of the planarization (equivalently the crossing number of $G$) is of primary importance in graph drawing, it is natural to ask what other properties can be transferred from $G$ to its planarizations.

One problem of this type arose in the work of Jansen and Wulms on the fixed-parameter tractability of graph optimization problems on graphs of bounded pathwidth~\cite{JanWul-IPEC-17}.
One of their constructions involved the planarization of a non-planar graph of bounded pathwidth, and they observed that the planarization maintained the low pathwidth of their graph.
Following this observation, Jansen asked on cstheory.stackexchange.com whether planarization preserves the property of having bounded pathwidth, and in particular whether $K_{3,n}$ (a graph of bounded pathwidth) has a bounded-pathwidth planarization.\footnote{See \url{https://cstheory.stackexchange.com/q/35974/95}.} This paper represents an extended response to this problem. We provide a negative answer to Jansen's question: planarizations of $K_{3,n}$ do not have bounded pathwidth. However, for bounded-degree graphs of bounded pathwidth, there always exists a planarization that maintains bounded pathwidth. More generally we study similar questions for many other standard graph width parameters.

Our work should be distinguished from a much earlier line of research on planarization and width, in which constraints on the width of planar graphs are transferred in the other direction, to information about the graph being planarized. In particular, Leighton~\cite{Lei-FOCS-81} used the facts that planar graphs have width at most proportional to the square root of their size, and that (for certain width parameters) planarization cannot decrease width, to show that when the original graph has high width it must have crossing number quadratic in its width. In our work, in contrast, we are assuming that the original graph has low width and we derive properties of its planarization from that assumption.

\subsection{Width parameters in graphs}

There has been a significant amount of research on graph width parameters and their algorithmic implications; see Ne{\v{s}}et{\v{r}}il and Ossona de Mendez~\cite{NesOss-SGSA-12} for a more detailed survey. We briefly describe the parameters that we use here.

\begin{figure}[t]
\centering\includegraphics[scale=0.5]{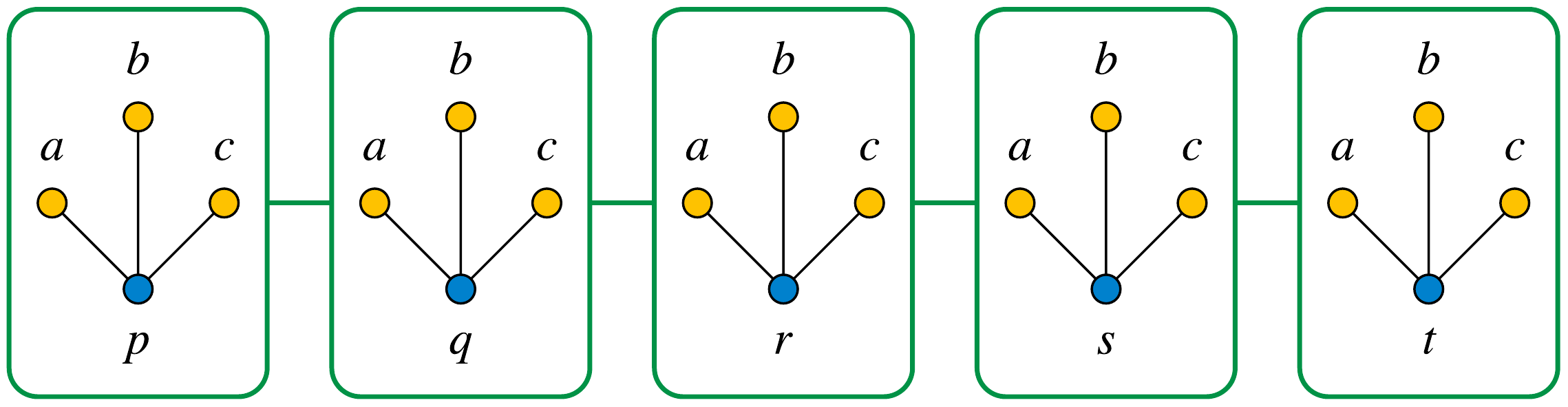}
\caption{A tree-decomposition and path-decomposition of $K_{3,5}$, with width three.
Vertices $a$, $b$, and $c$ (on one side of the bipartition) belong to all bags; vertices $p$, $q$, $r$, $s$, and $t$ (on the other side) are each in only one bag.}
\label{fig:pd-K-3-5}
\end{figure}

\begin{description}
\item[Treewidth.]  Treewidth has many equivalent definitions; the one we use is that the treewidth of a graph $G$ is the minimum width of a tree-decomposition of $G$~\cite{NesOss-SGSA-12}. Here, a tree-decomposition is a tree $T$ whose nodes, called \emph{bags}, are labeled by sets of vertices of $G$. Each vertex of $G$ must belong to the bags of a contiguous subtree of $T$, and for each edge of $G$ there must exist a bag containing both endpoints of the edge. The width of the decomposition is one less than the maximum cardinality of the bags. \autoref{fig:pd-K-3-5} shows such a decomposition for $K_{3,5}$.

\item[Pathwidth.] The pathwidth of a graph is the minimum width of a tree-decom\-po\-si\-tion whose tree is a path, as it is in \autoref{fig:pd-K-3-5}~\cite{NesOss-SGSA-12}.  Equivalently the pathwidth equals the minimum vertex separation number of a linear arrangement of the vertices of $G$ (an arrangement of the vertices into a linear sequence)~\cite{Kin-IPL-92}. Every linear arrangement of an $n$-vertex graph defines $n-1$ cuts, that is, $n-1$ partitions of the vertices into a prefix of the sequence and a disjoint suffix of the sequence. The vertex separation number of a linear arrangement is the maximum, over these cuts, of the number of vertices in the prefix that have a neighbor in the suffix.
From a linear arrangement one can construct a tree-decomposition in the form of a path, where the first bag on the path for each vertex $v$ contains $v$ together with all vertices that are earlier than $v$ in the arrangement but that have $v$ or a later vertex as a neighbor.

\item[Cutwidth.] The cutwidth of a graph equals the minimum edge separation number of a linear arrangement of the vertices of $G$~\cite{ChuSey-DM-89}. The edge separation number of a linear arrangement is the maximum, over the prefix--suffix cuts of the arrangement, of the number of edges that cross the cut.

\item[Bandwidth.] The bandwidth of a graph equals the minimum span of a linear arrangement of the vertices of $G$~\cite{ChuSey-DM-89}. The span of a linear arrangement is the maximum, over the edges of $G$, of the number of steps in the linear arrangement between the endpoints of the edge.

\item[Branchwidth.] A branch-decomposition of a graph $G$ is an undirected tree $T$, with leaves labeled by the edges of $G$, and with every interior  vertex of $T$ having degree three. Removing any edge $e$ from $T$ partitions $T$ into two subtrees; these subtrees partition the leaves of $T$ into two sets, and correspondingly partition the edges of $G$ into two subgraphs. The width of the decomposition is the maximum, over all edges $e$ of $T$, of the number of vertices that belong to both subgraphs. The branchwidth of $G$ is the minimum width of any branch-decomposition~\cite{SeyTho-Comb-94}.

\item[Carving width.] A carving decomposition of a graph $G$ is an undirected tree $T$, with leaves labeled by the vertices of $G$, and with every interior  vertex of $T$ having degree three. Removing any edge $e$ from $T$ partitions $T$ into two subtrees; these subtrees partition the leaves of $T$ into two sets, and correspondingly partition the vertices of $G$ into two induced subgraphs. The width of the decomposition is the maximum, over all edges $e$ of $T$, of the number of edges of $G$ that connect one of these subgraphs to the other. The carving width of $G$ is the minimum width of any carving decomposition~\cite{SeyTho-Comb-94}.
For instance, \autoref{fig:carve-K-3-3} depicts a carving decomposition of $K_{3,3}$ with width four, the minimum possible for this graph.

\begin{figure}[t]
\centering\includegraphics[scale=0.5]{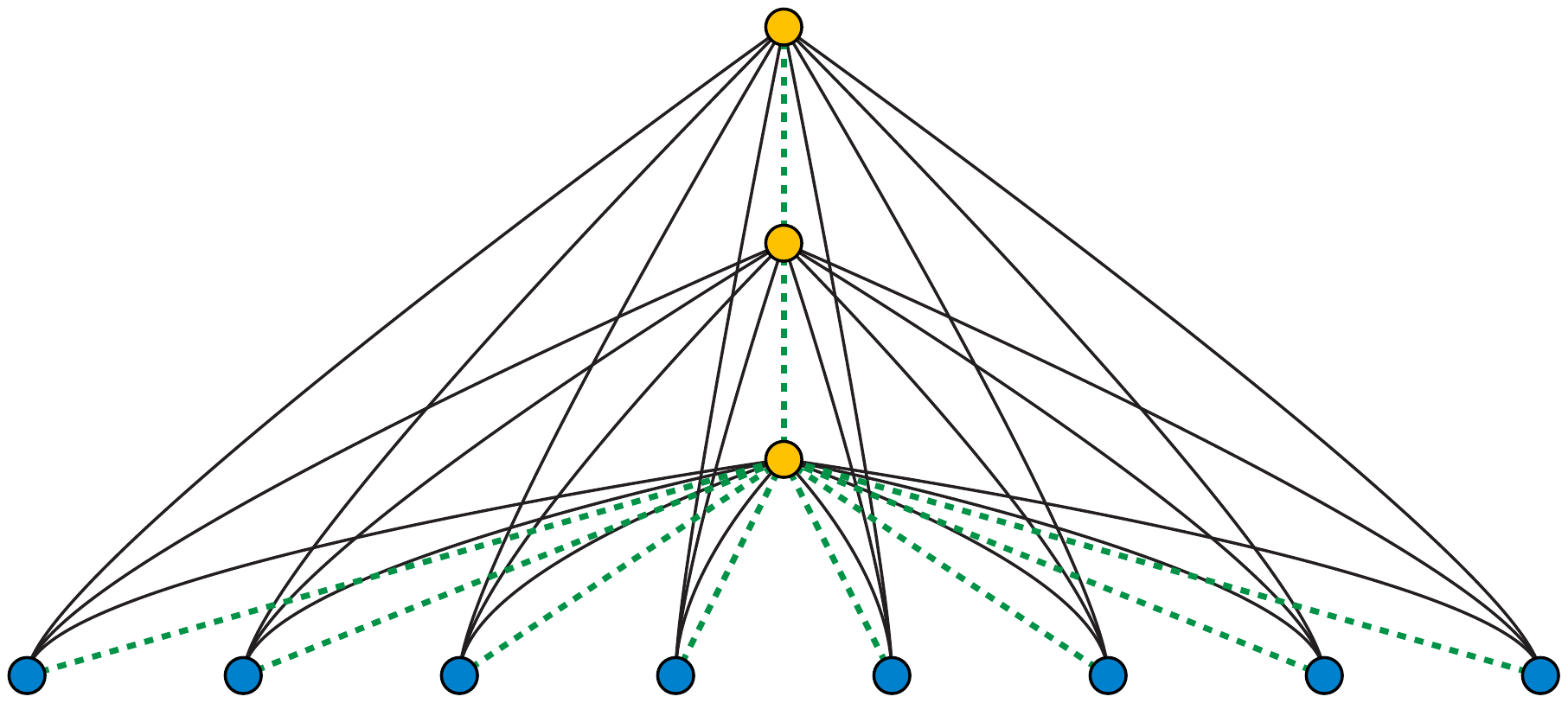}
\caption{$K_{3,8}$ has tree-depth three: The depth-three tree shown by the green dashed edges forms a depth-first search tree of a supergraph of $K_{3,8}$.}
\label{fig:td-K-3-8}
\end{figure}

\item[Tree-depth.] The tree-depth of $G$ is the minimum depth of a depth-first-search tree $T$ of a supergraph of $G$ (\autoref{fig:td-K-3-8}). Such a tree can be characterized more simply by the property that every edge of $G$ connects an ancestor--descendant pair in~$T$~\cite{NesOss-SGSA-12}.
\item[Clique-width.] A clique-construction of a graph $G$ is a process that constructs a vertex-colored copy of $G$ from smaller vertex-colored graphs by steps that create a new colored vertex, take the disjoint union of two colored graphs, add all edges from vertices of one color to vertices of another, or assigning a new color to vertices of a given color. The width of a clique-construction is the number of distinct colors it uses, and the clique-width of a graph is the minimum width of a clique-construction~\cite{CorRot-SICOMP-05}.
\end{description}

\subsection{New results}

In this paper, we consider for each of the depth parameters listed above how the parameter can change from a graph to its planarization, when the planarization is chosen to minimize the parameter value.
We show that for treewidth, pathwidth, branchwidth, tree-depth, and clique-width there exists a graph with bounded parameter value, all of whose planarizations have parameter value $\Omega(n)$. In each of these cases, the graph can be chosen as a complete bipartite graph $K_{3,n}$. (It was also  known that the planarizations of $K_{3,n}$ have quadratic size~\cite{Zar-FM-54}.)

However, for bandwidth, cutwidth, and carving width, every graph with bounded parameter value has a planarization of linear size whose parameter value remains bounded. The same is true for the treewidth, pathwidth, branchwidth, and clique-width of graphs of bounded degree. (Graphs of bounded degree and bounded tree-depth have bounded size, so this final case is not interesting.)

\section{Treewidth, branchwidth, pathwidth, tree-depth, and clique-width}

In this section we show that all planarizations of $K_{3,n}$ have high width.
We begin by computing the crossing number of $K_{3,n}$. This is a special case of Tur\'an's brick factory problem of determining the crossing number of all complete bipartite graphs.
For our results we need a variant of the crossing number, $\paircr(G)$, defined as the minimum number of pairs of crossing edges (allowing edges to cross each other multiple times, but only counting a single crossing in each case) instead of the usual crossing number $\crn(G)$
defined as the number of points where edges cross~\cite{PacTot-FOCS-98}.
We bound this number by an adaptation of an argument from Kleitman~\cite{Kle-JCT-70}, who credits it to Zarankiewicz~\cite{Zar-FM-54}.

\begin{lemma}
\label{lem:cr-K-3-n}
\[
\paircr(K_{3,n})=
\binom{\lfloor n/2\rfloor}{2}+
\binom{\lceil n/2\rceil}{2}=
\biggl\lfloor\frac{n}{2}\biggr\rfloor
\biggl\lfloor\frac{n-1}{2}\biggr\rfloor.
\]
\end{lemma}

\begin{figure}[t]
\centering\includegraphics[scale=0.5]{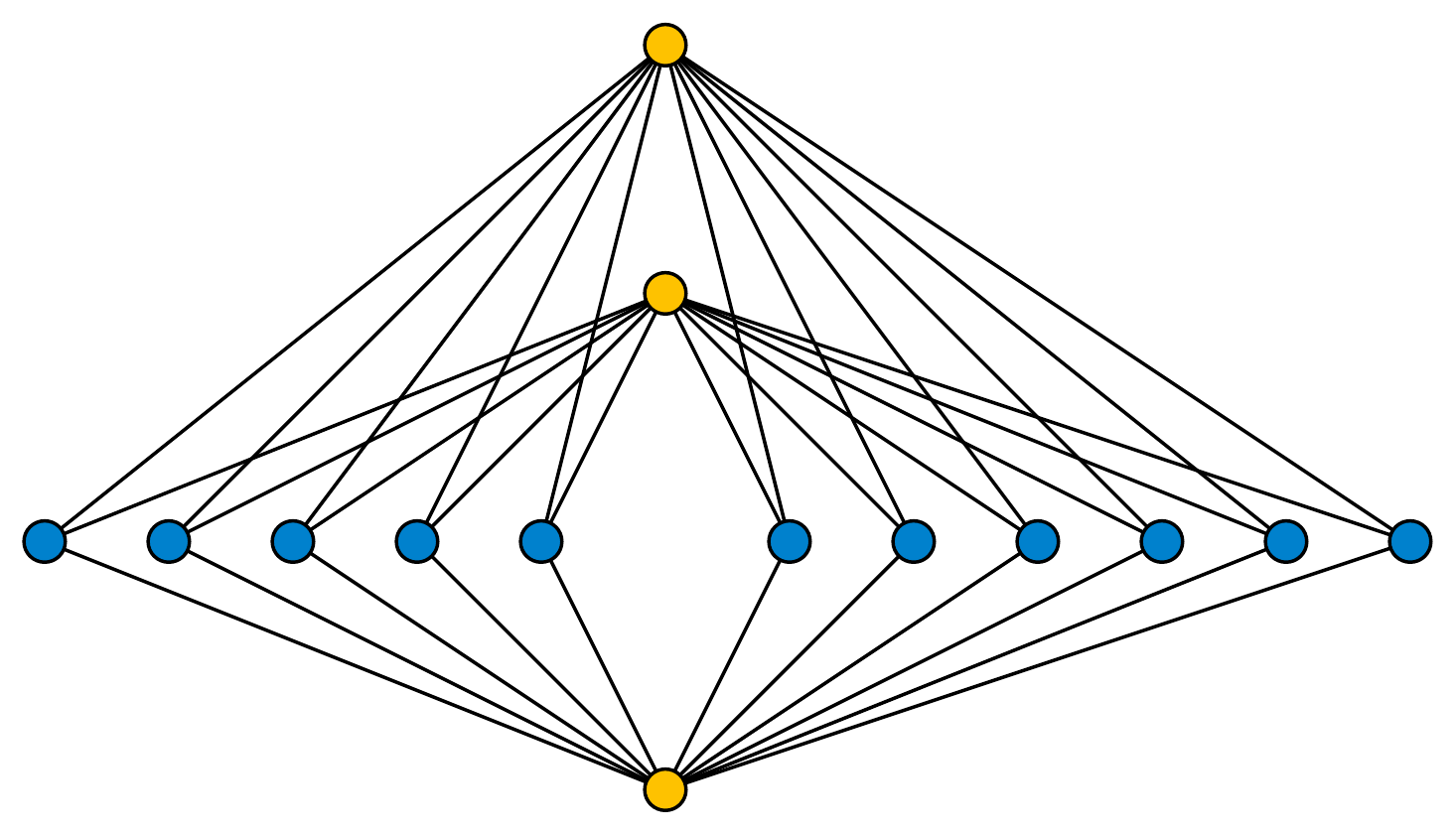}
\caption{A drawing of $K_{3,11}$ with $25$ crossings, the minimum possible for this graph.}
\label{fig:cr-K-3-11}
\end{figure}

\begin{proof}
To show that a drawing with this many crossing pairs exists, place the $n$ vertices on one side of the bipartition of $K_{3,n}$ along the $x$-axis, with $\lfloor n/2\rfloor$ on one side of the origin and $\lceil n/2\rceil$ on the other. Place the three vertices on the other side of the bipartition along the $y$-axis, with two points on one side of the origin and one on the other. Connect all of the pairs of points that have one point on each axis by a straight line segment, as shown in \autoref{fig:cr-K-3-11}. A straightforward calculation shows that the number of crossings is as claimed.

In the other direction, we know as base cases that $\paircr(K_{3,2})=0$ and $\paircr(K_{3,3})=1$. For any larger $n$, let the vertices of the $n$-vertex side of the bipartition of $K_{3,n}$ be $v_1,v_2,\dots v_n$. If every pair $v_i,v_j$ form the endpoints of at least one pair of crossing edges, then the total number of crossings is at least $\tbinom{n}{2}$, larger than the stated bound; otherwise, order the vertices so that $v_{n-1}$ and $v_n$ do not form the endpoints of any pair of crossing edges.

Then, in any drawing of $K_{3,n}$, the $K_{3,n-2}$ subgraph formed by deleting $v_{n-1}$ and $v_n$ has $\paircr(K_{3,n-2})$ crossings. Each of the $n-2$ $K_{3,3}$ subgraphs induced by $v_{n-1}$, $v_{n}$, exactly one other $v_i$, and the three vertices on the other side of the bipartition supplies at least one additional crossing, because $\paircr(K_{3,3})=1$. None of these subgraphs share any crossings, because the crossings in the $K_{3,n-2}$ subgraph involve neither $v_{n-1}$ nor $v_n$, while the crossings in the $K_{3,3}$ subgraph all involve exactly one of these two vertices and the one other vertex $v_i$ included in the subgraph. Therefore, we have that
\[
\paircr(K_{3,n}) \ge
\paircr(K_{3,n-2}) + (n-2)\paircr(K_{3,3}).
\]
The result follows by induction on $n$.
\end{proof}

This lemma shows that the \emph{crossing graph} of a drawing, with a vertex in the crossing graph for each edge of $K_{3,n}$ and an edge in the crossing graph for each crossing of the drawing, has constant density, in the following sense. We define the \emph{density} of a graph with $m$ edges and $n$ vertices to be $m/\tbinom{n}{2}$.
This is a number in the range $[0,1]$. For instance, the crossing graph of any planarization of $K_{3,n}$ has $3n$ vertices and (by \autoref{lem:cr-K-3-n}) at least $\bigl(1-o(1)\bigr)n^2/4$ edges, so its density is at least
\[
\bigl(1-o(1)\bigr) \frac{n^2}{4} \Bigm/ \binom{3n}{2} = \frac{1}{18}-o(1).
\]
To prove that planarizations of $K_{3,n}$ have high treewidth, we need higher densities than this, which we will achieve using the following ``densification lemma'':

\begin{lemma}
\label{lem:densification}
Let $G$ be a disconnected graph with $n$ vertices and $m$ edges, such that the $i$th connected component of $G$ has $n_i$ vertices and $m_i$ edges. Then there exists $i$ such that $m_i/n_i\ge m/n$.
\end{lemma}

\begin{proof}
We can represent $m/n$ as a convex combination of the corresponding quantities in the subgraphs:
\[
\frac{m}{n}=\sum_i \frac{n_i}{n}\cdot\frac{m_i}{n_i}.
\]
The result follows from the fact that a convex combination of numbers cannot exceed the maximum of the numbers.
\end{proof}

Given a tree-decomposition $T$ of a drawing $D$ of $K_{3,n}$, and any connected subtree $S$ of $T$, we define a crossing graph $C_S$ as follows. Let $E_S$ be the subset of edges of $K_{3,n}$ with the property that, for each edge $e$ in $E_S$, the only bags of $T$ that contain crossings on $e$ are the bags in $S$.  (We do \emph{not} require the two endpoints of $e$ in $K_{3,n}$ to belong to these bags.) Then $C_S$ is a graph having $E_S$ as its vertex set, and having an edge for each pair of edges in $E_S$ that cross in~$D$. For instance, $C_T$ is the crossing graph of the whole drawing, as defined earlier. We will use \autoref{lem:densification} to find subtrees $S$ whose graphs $C_S$ are more dense (relative to their numbers of vertices) than $C_T$. To do so, we use the separation properties of trees:

\begin{lemma}
\label{lem:td-separator}
Let $T$ be a tree-decomposition $T$ of a drawing $D$ of $K_{3,n}$,
and suppose that $T$ has width $w$.
Let $S$ be a subtree of $T$ such that $C_S$ has $n_S$ vertices and $m_S$ edges.
Then there exists a bag $B\in S$ with the following properties:
\begin{itemize}
\item The removal of $B$ disconnects $S$ into subtrees $S_i$.
\item For subtree $S_i$, the corresponding crossing graph $C_{S_i}$ has at most $n_S/2$ vertices.
\item The total number of edges in all of the crossing graphs $C_{S_i}$ for all of the subtrees $S_i$
is at least $S-2(w+1)(n-2)$.
\end{itemize}
\end{lemma}

\begin{proof}
We choose $B$ arbitrarily, and then as long as it fails to meet the condition on the number of vertices in the graphs $C_{S_i}$ we move $B$ to the (unique) adjacent bag in which this condition is not met.
After the move, the subtree containing the former location of $B$ has at most $n_S/2$ vertices in $C_{S_i}$, because these vertices are disjoint from the ones in the large subtree before the move.
Moving $B$ also cannot increase the numbers of vertices in the crossing graphs of the subtrees formed from the partition of the large subtree, and it reduces the numbers of bags in those subtrees.
Therefore, this process must eventually terminate at a choice of $B$ for which all crossing graphs have the stated number of vertices.

An edge~$e$ in $C_S$ (representing a crossing between two edges $f$ and $f'$ of $K_{3,n}$) will belong to one of the $C_{S_i}$ unless $B$ contains a crossing point on~$f$ or on~$f'$. $B$ may contain at most $w+1$ crossings of $D$, and each may eliminate at most $2(n-2)$ edges of $C_S$ (if it is a crossing of two edges in $E_S$ and each has $n-2$ other crossings). Therefore, the total number of edges in all of the crossing graphs $C_{S_i}$ for all of the subtrees $S_i$ is as stated.
\end{proof}

\begin{theorem}
\label{thm:planarize-K-3-n}
Every planarization of $K_{3,n}$ has treewidth $\Omega(n)$.
\end{theorem}

\begin{proof}
Let $D$ be an arbitrary planarization of $K_{3,n}$, and let $T$ be a minimum-width tree-decomposition of $D$. Let $\epsilon>0$ be a constant to be determined later.
We will show that $T$ either has width at least $\epsilon n$, or it has a subtree $S$ whose crossing graph $C_S$ has density strictly greater than one. Since no graph (without repeated edges) can have density so high, the only possibility is that $T$ has width at least $\epsilon n=\Omega(n)$.

To find $S$, for drawings whose width is at most $\epsilon n$, begin with $S=T$. Then, repeatedly use \autoref{lem:td-separator} to partition the current choice of subtree $S$ into smaller subtrees, and then use \autoref{lem:densification} to find one of these subtrees that is dense.
Each such step reduces the number of vertices in $C_S$ by at least a factor of two,
while also reducing the number of edges by approximately the same reduction factor
(approximately because of the $O(\epsilon n^2)$ edges of the crossing graph that are eliminated by the application of \autoref{lem:td-separator}).
Therefore, each step increases the density of $C_S$ by a factor of $2-O(\epsilon)$.
When $S=T$, the density is at least $1/18-o(1)$, so after at most five steps the density is $(32-O(\epsilon))(1/18-o(1))$.

To complete the argument, we need only choose $\epsilon$ to be small enough so that this expression, $(32-O(\epsilon))(1/18-o(1))$, has a value exceeding one.
\end{proof}

\begin{corollary}
\label{cor:all-params-linear}
For every planarization of $K_{3,n}$, and every parameter in $\{$pathwidth, cutwidth, bandwidth, branchwidth, carving width, tree-depth, clique-width$\}$, the value of the parameter on the planarization is $\Omega(n)$. Therefore, there exists a family of graphs for which each of these parameters is bounded but for each each planarization has linear parameter value.
\end{corollary}

\begin{proof}
All of these parameters except clique-width are bounded from below by a linear function of the treewidth.

\begin{figure}[t]
\centering\includegraphics[scale=0.5]{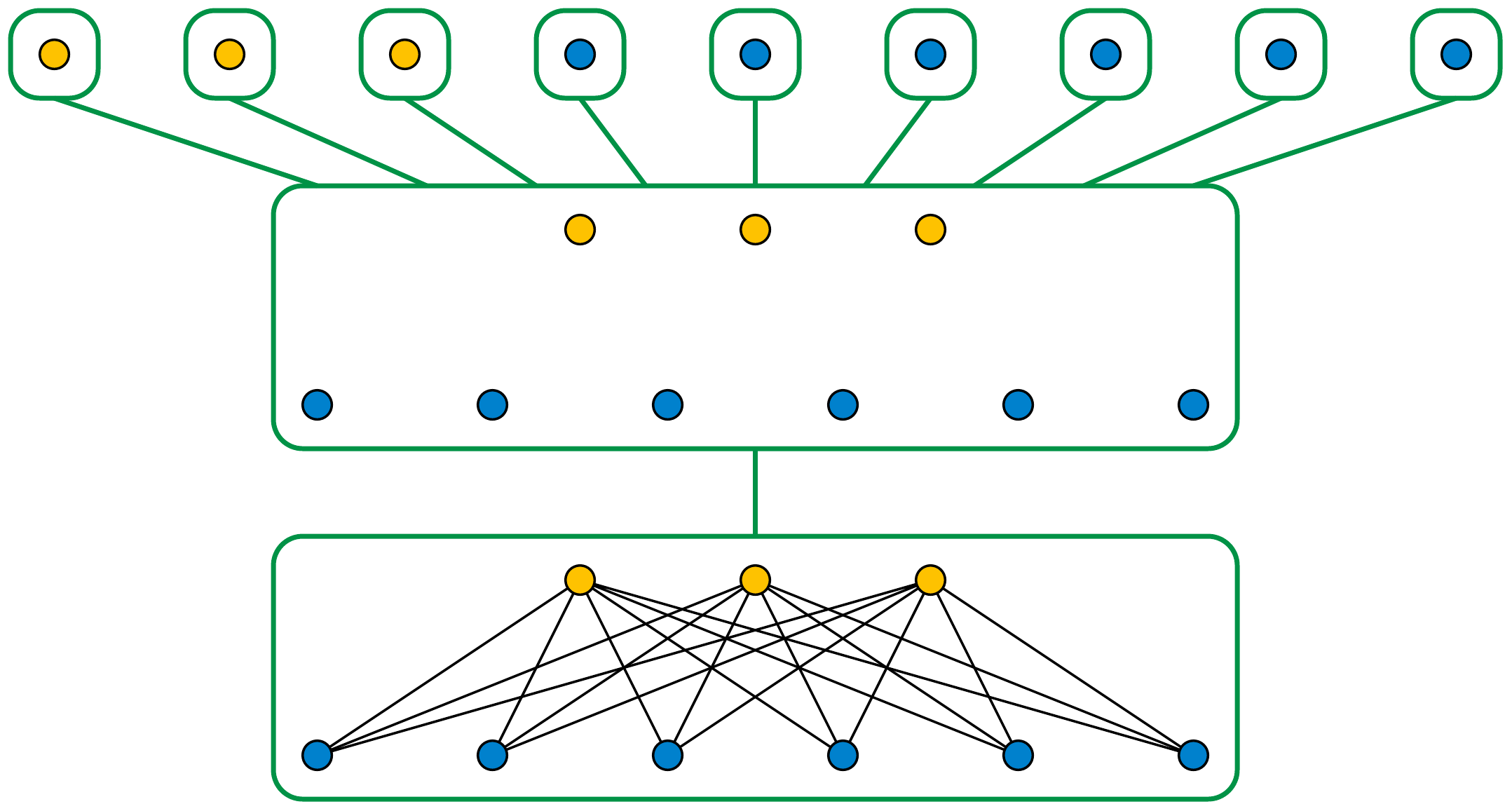}
\caption{Clique-width 2 construction of $K_{3,6}$ by a disjoint union of colored single vertices, followed by an operation that adds an edge between each bichromatic pair of vertices.}
\label{fig:kw-k-3-6}
\end{figure}

As with any complete bipartite graph, the clique-width of $K_{3,n-3}$ is two: it can be constructed from a  disjoint union of single vertices of two colors, by adding edges between all bichromatic pairs of vertices (\autoref{fig:kw-k-3-6}). The $\Omega(n)$ lower bound on clique-width follows from the facts that (as a planar graph) any planarization has no $K_{3,3}$ subgraph and that, for graphs with no $K_{t,t}$ subgraph, the treewidth is upper-bounded by a constant factor (depending on~$t$) times the clique-width~\cite{GurWan-WG-00}. Equivalently, the clique-width of any planarization is lower-bounded
by a constant times its treewidth, which by  \autoref{thm:planarize-K-3-n} is $\Omega(n)$.
\end{proof}

\section{Cutwidth and bounded-degree pathwidth}

Cutwidth behaves particularly well under planarization:

\begin{theorem}
\label{thm:cutwidth}
Let $G$ be a graph with $n$ vertices and $m$ edges, of cutwidth $w$. Then $G$ has a planarization
with $O(n+wm)$ vertices, of cutwidth at most~$w$.
\end{theorem}

\begin{proof}
Consider a linear arrangement of $G$ with edge separation number $w$, and use the positions in this arrangement as $x$-coordinates for the vertices. Assign the vertices $y$-coordinates that place them into convex and general position, draw the edges of $G$ as straight line segments between the resulting points, and planarize the drawing by replacing each crossing by a vertex. Here, by ``general position'' we mean that no two points have the same $x$-coordinate,
no five points form a pentagon in which two crossing points and a vertex have the same $x$-coordinate,
no six points form a hexagon with three coincident diagonals,
and no eight points form an octagon in which the crossing points of two pairs of diagonals have the same $x$-coordinate. This will all be true of a rotation by a sufficiently small but nonzero angle of any convex placement. In the resulting drawing, there can be no intersections of vertices or edges other than incidences and simple crossings, and no two vertices or crossing points can have the same $x$-coordinate. An example is shown in \autoref{fig:cw-K-3-4}.

We use the ordering by $x$-coordinates of the planarization as a linear arrangement of the planarization. The edge intersection number is the maximum number of edges in the drawing that can be cut by any vertical line, unchanged between $G$ and its planarization.

\begin{figure}[t]
\centering\includegraphics[scale=0.5]{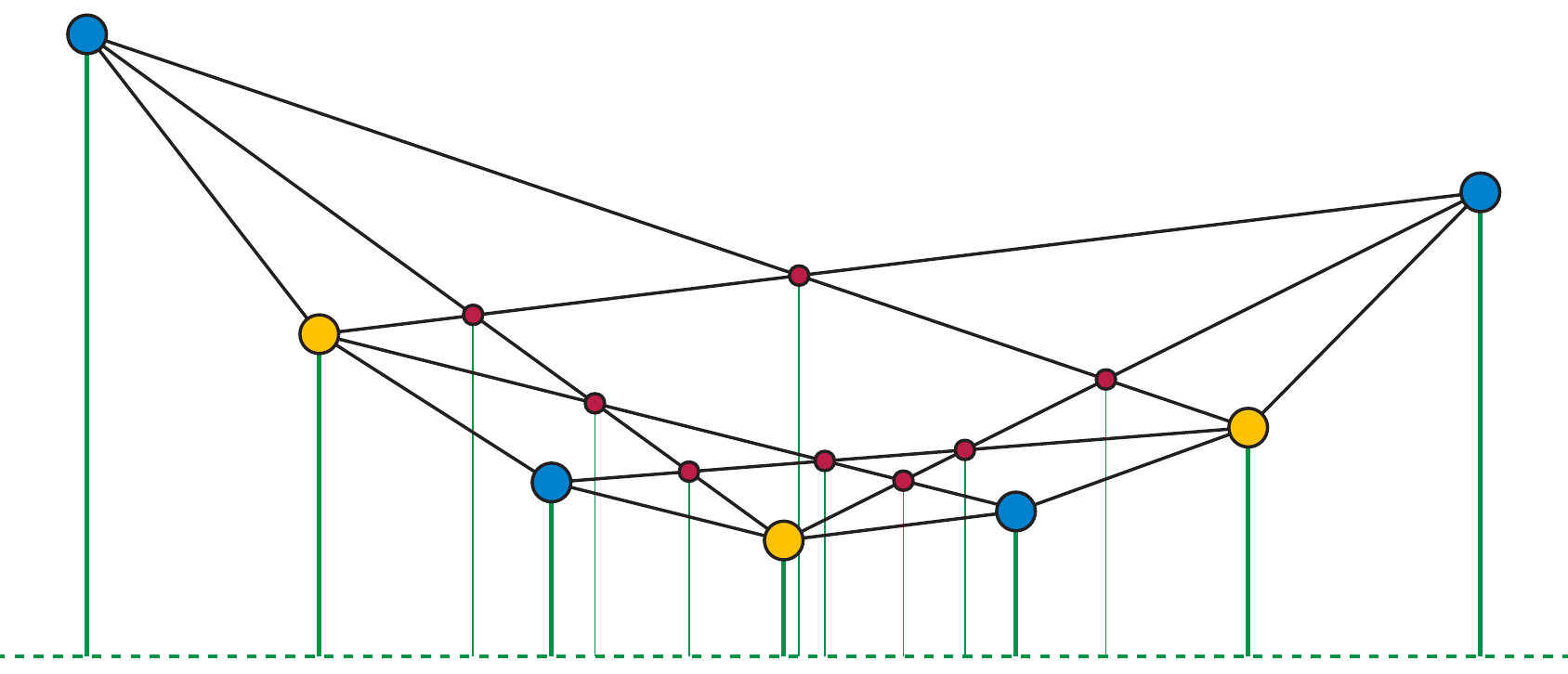}
\caption{Planarizing a graph of low cutwidth (here $K_{3,4}$, drawn with edge separation number six) by lifting its linear arrangement to a convex curve.}
\label{fig:cw-K-3-4}
\end{figure}

Because of the convex position of the vertices of $G$, each edge $(u,v)$ of $G$ can only be crossed by other edges that cross exactly one of the two vertical lines through $u$ and $v$; there are $O(w)$ such edges, so the number of crossings per edge is $O(w)$ and the total number of crossings is $O(wm)$.
\end{proof}

The lower bound of \autoref{cor:all-params-linear} does not contradict \autoref{thm:cutwidth} because $K_{3,n}$ does not have bounded cutwidth. Its cutwidth is at least $3\lceil n/2\rceil$,
obtained in any linear arrangement at the cut between the first $\lceil n/2\rceil$ vertices on the $n$-vertex side of the bipartition (together with any vertices from the other side that are mixed among them) and the remaining vertices of the graph.
For instance, the drawing of $K_{3,4}$ in \autoref{fig:cw-K-3-4} achieves the optimal cutwidth of six for this graph.
An example showing that the planarization size bound is tight is given by a disjoint union of $O(n/w)$ bounded-degree expander graphs, each having $O(w)$ vertices and crossing number $\Theta(w^2)$.

\begin{corollary}
Let $G$ be a graph with bounded pathwidth and bounded maximum degree. Then $G$ has a planarization with linear size and bounded pathwidth.
\end{corollary}

\begin{proof}
If a graph has pathwidth $w$ and maximum degree $d$, it has cutwidth at most $dw$~\cite{ChuSey-DM-89}, and so does its planarization (\autoref{thm:cutwidth}). Because the planarization has cutwidth at most $dw$, it also has pathwidth at most $dw$, because the vertex separation number of any linear arrangement is at most equal to the edge separation number (with equality when the separation number is achieved by a matching).
\end{proof}

\section{Bandwidth}

The same construction used for planarizing graphs with low cutwidth also works for graphs of low bandwidth.

\begin{theorem}
Let $G$ be a graph with $n$ vertices and $m$ edges, of bandwidth $w$. Then $G$ has a planarization
with $O(n+w^2m)$ vertices, whose bandwidth is $O(w^4)$.
\end{theorem}

\begin{proof}
We lift a linear arrangement of $G$ with low span to a convex curve in the plane, as in the proof of 
\autoref{thm:cutwidth}. Within the span of any edge $e$ of $G$, there are $O(w^2)$ other edges and $O(w^4)$ crossings of those edges, so the span of $e$ in the planarization is $O(w^4)$. This bound applies also to the span of any segment of $e$ created by crossings with other vertices. Each edge may be crossed by $O(w^2)$ other edges, so the total number of dummy vertices added is $O(w^2m)$.
\end{proof}

It may be possible to reduce the bandwidth of the planarization by introducing artificial crossings to break up edges with long spans, but we have not pursued this approach as we do not believe it will lead to better graph drawings.

\section{Carving width and bounded-degree treewidth}

If a graph has low carving width, we can use its carving decomposition (a tree with the vertices at its leaves, internal degree three, and with few edges spanning the cut determined by each tree edge) to guide a drawing of the algorithm that leads to a planarization with low carving width.

It is helpful, for our construction, to relate carving width to cutwidth.

\begin{lemma}
\label{lem:cutwidth2carving}
If a graph $G$ has cutwidth $w$ and maximum degree $d$, then $G$ has carving width at most $\max(w,d)$.
\end{lemma}

\begin{proof}
We form a carving decomposition of $G$ in the form of a caterpillar: a path with each path vertex having a single leaf connected to it (except for the ends of the path which have two connected leaves). The ordering of the leaves is given by a linear arrangement minimizing the edge separation number. Then the cuts of the carving decomposition that are determined by edges of the path are exactly the ones determining the edge separation number, $w$. The remaining cuts, determined by leaf edges of the tree, are crossed by the neighboring edges of each vertex, of which there are at most~$d$. An example of this construction can be seen in \autoref{fig:cw-K-3-4}: the dashed horizontal green line represents the path from which the carving decomposition is formed, the heavy vertical green lines correspond to the leaf edges of the carving decomposition of $K_{3,4}$, and the thin vertical green edges correspond to the leaf edges of the carving decomposition of a planarization of $K_{3,4}$.
\end{proof}

\begin{theorem}
\label{thm:carving}
If an $n$-vertex graph $G$ has carving width $w$, then $G$ has a planarization with $O(w^2 n)$ additional vertices that still has carving width at most~$w$.
\end{theorem}

\begin{figure}[t]
\centering\includegraphics[scale=0.5]{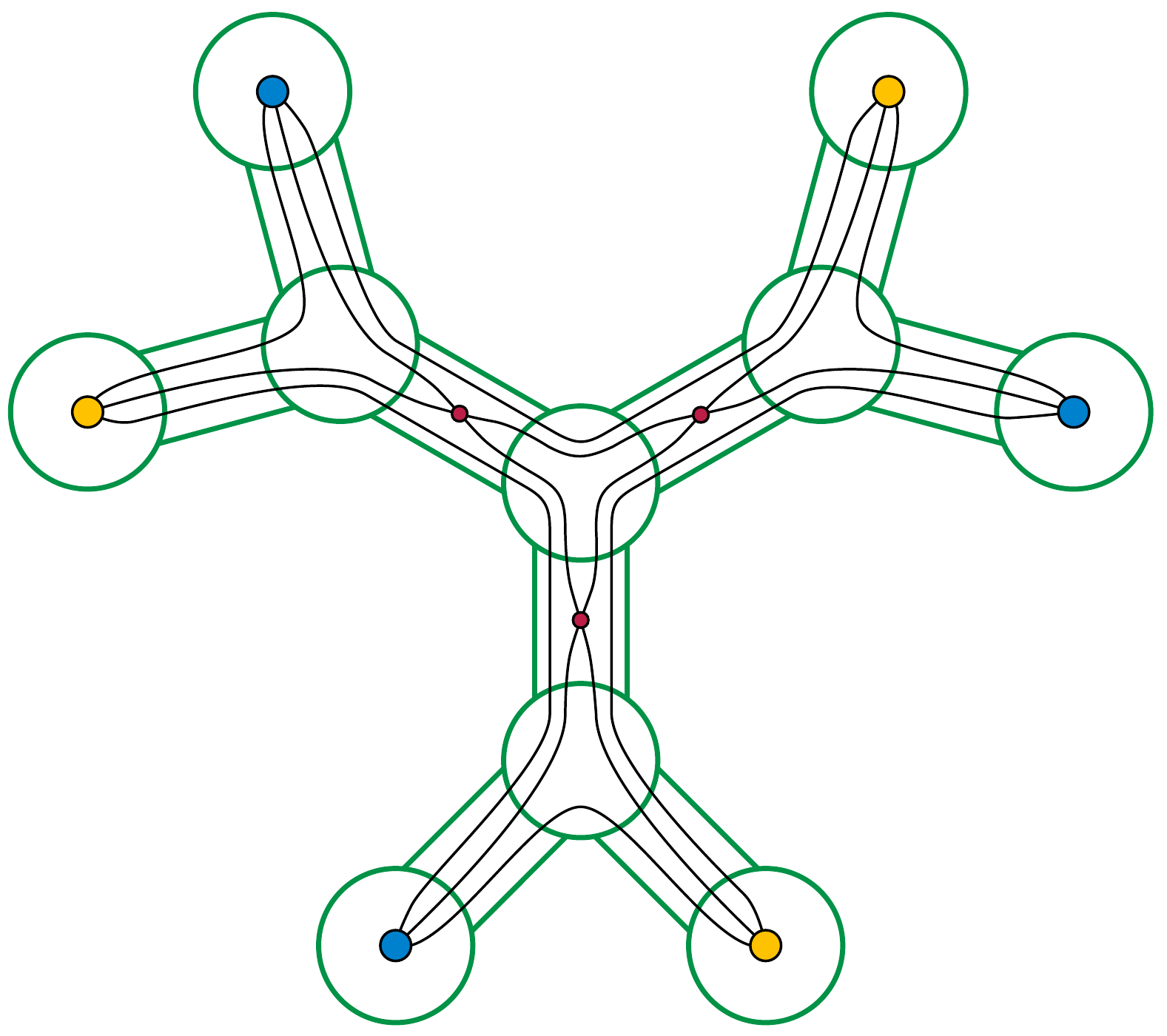}
\caption{Using a carving decomposition of $K_{3,3}$ to guide a planarization.}
\label{fig:carve-K-3-3}
\end{figure}

\begin{proof}
Let $T$ be the tree of a carving decomposition of $G$ with width $w$. Draw $T$ without crossings in the plane, with straight-line edges, and thicken the vertices of $T$ to disks and the edges of $T$ to rectangles without introducing any additional self-intersections of the drawing.
Place each vertex of $G$ in the disk of the corresponding leaf vertex of $T$. Route each edge of $G$ as a curve through the rectangles and disks connecting its endpoints, so that within each rectangle it forms a monotone curve (with respect to the orientation of the rectangle) crossing at most once each other edge routed within the same rectangle, and so that, at each end of each rectangle, the curves are sorted by the ordering of their destination leaves in the planar embedding of $T$. With this sorted ordering, there need not be any crossings within the disks representing internal vertices of $T$, nor in the rectangles representing leaf edges of~$T$ (\autoref{fig:carve-K-3-3}). The $n-3$ remaining edges of $T$ each contain at most $\tbinom{w}{2}$ crossings.  So the total number of crossings is at most $(n-3)\tbinom{w}{2}=O(w^2n)$.

This drawing cannot yet be recognized as a carving decomposition of a planarization of $G$, because some of its vertices (the dummy vertices introduced at crossings) are now placed along the edges of $T$ rather than at leaves. However, by topologically sweeping the arrangement of monotone curves~\cite{EdeGui-JCSS-89} within each rectangle corresponding to an edge of~$T$,
we can arrange the crossing points within that rectangle into a linear sequence, such that the portion of the drawing within that rectangle has edge separation number at most $w$ for that sequence.
Applying \autoref{lem:cutwidth2carving} (replacing the edge of $T$ by a carving decomposition in the form of a caterpillar, with a leaf of the decomposition for each vertex added in the planarization to replace a crossing of $G$, and with the ordering of these leaves given by a topological sweep of the arrangement) produces a carving decomposition of the planarization with width~$w$, as required.
\end{proof}

We note that this planarization technique resembles the ``simple planarization'' method of Di Battista et al.~\cite{DiBDidMar-GD-01} for clustered graphs. In this respect, we may view the carving decomposition of $G$ as a clustering to be respected by the planarization.

In an appendix, we prove the following strengthening of \autoref{thm:carving}:

\begin{theorem}
\label{thm:stronger-carving}
If an $n$-vertex graph $G$ has carving width $w$, then $G$ has a planarization with $O(w^{3/2} n)$ additional vertices that still has carving width $O(w)$.
\end{theorem}

An example showing that \autoref{thm:stronger-carving} is tight is given by a cluster graph consisting of $O(n/\sqrt{w})$ disjoint cliques of size $O(\sqrt{w})$, each requiring $\Theta(w^2)$ crossings in any drawing.

\begin{corollary}
Let $G$ be a graph with bounded treewidth or branchwidth and bounded maximum degree. Then $G$ has a planarization with linear size and bounded treewidth and branchwidth.
\end{corollary}

\begin{proof}
Treewidth and branchwidth are always within a constant factor of each other~\cite{SeyTho-Comb-94} so we may concentrate on the results for branchwidth, and the corresponding results for treewidth will follow automatically.

A carving decomposition may be converted into a branch decomposition by replacing each leaf of the carving decomposition (representing a vertex of the given graph) with a subtree (representing edges adjacent to the given vertex), in such a way that each edge is represented at exactly one of its endpoints. This increases the width of the decomposition by at most a factor equal to the degree. In the other direction, a branch decomposition may be converted into a carving decomposition by replacing each leaf of the branch decomposition (representing an edge of the given graph) by a subtree of zero, one, or two leaves (representing endpoints of the edge) in such a way that each vertex is represented at exactly one of its incident edges.This increases the width of the decomposition by at most a factor of two. So, the carving width is at most the degree times the branchwidth, and at least half the branchwidth~\cite{NesThi-DAM-14}.

Therefore, if $G$ has bounded branchwidth and bounded maximum degree, it has bounded carving width, and \autoref{thm:carving} tells us that it has a planarization of linear size that also has bounded carving width. The same planarization also must have bounded branchwidth.
\end{proof}

\iffull
\bibliographystyle{unabuser}
\else
\bibliographystyle{splncs}
\fi
\bibliography{pwidth}

\iffull
\newpage\appendix

\section{Tight crossing bounds for carving width}

Recall that \autoref{thm:stronger-carving} states that, if an $n$-vertex graph $G$ has carving width $w$, then $G$ has a planarization with $O(w^{3/2} n)$ additional vertices that still has carving width $O(w)$. The example of the disjoint union of $O(n/\sqrt{w})$ cliques of $O(\sqrt{w})$ vertices shows that this is tight.

To prove \autoref{thm:stronger-carving}, we apply a tree clustering technique of Frederickson~\cite{Fre-SICOMP-97} to the carving decomposition of $G$.

Following Frederickson~\cite{Fre-SICOMP-97} we define a \emph{restricted partition of order $z$} of an unrooted binary tree $T$ (such as the tree of a carving decomposition) to be a partition of the vertices of $T$ into connected subtrees with the following properties:
\begin{itemize}
\item Each subtree of the partition contains at most $z$ vertices.
\item If a subtree of the partition has more than two edges connecting it to other subtrees, then it contains exactly one vertex.
\item If two subtrees of the partition are connected by an edge, then they cannot be merged into a single subtree while preserving the previous two properties.
\end{itemize}

Such a partition can be found easily by a greedy algorithm that repeatedly merges subtrees until no more merges are possible.

\begin{lemma}[Frederickson~\cite{Fre-SICOMP-97}]
For every unrooted binary tree $T$ with $n$ vertices, every $z$, and every restricted partition of $T$ of order $z$, there are at most $O(n/z)$ subtrees in the partition.
\end{lemma}

\begin{proof}
 If each subtree in a restricted partition is contracted into a single vertex, the result is again an unrooted tree with maximum degree three. Every leaf vertex of this contracted tree together with its parent must together have more than $z$ vertices, or else they could be merged to form a larger tree with at most $z$ vertices and at most two connecting edges to other subtrees. For the same reason, every pair of adjacent degree-two vertices in this contracted tree must together have more than $z$ vertices. Therefore, the contracted tree can only have $O(n/z)$ leaf vertices and $O(n/z)$ adjacent pairs of degree-two vertices, from which it follows that it has $O(n/z)$ vertices altogether.
 \end{proof}
 
To planarize $G$ with $O(w^{3/2} n)$ additional vertices and carving width $O(w)$, proving \autoref{thm:stronger-carving}, we first find a restricted partition of the carving decomposition $T$ of $G$, of order $O(\sqrt{w})$.

Each subtree $T_i$ of the restricted partition represents a subset of $O(\sqrt n)$ vertices of $G$, possibly having up to $2w$ edges connecting it to the rest of $G$ along the two edges of $T$ connecting this subtree to the rest of $T$. Let $V_i$ denote the subset of vertices of $G$ within subtree $T_i$, together with up to two dummy vertices representing the two edges of $T$ connecting $T_i$ to the rest of $T$. We planarize the subgraph of edges that enter or pass through $T_i$ by placing the vertices of $V_i$ onto a circle, but otherwise in general position, and by drawing each edge as a straight line segment between two points of this circle. Each of the $O(n/\sqrt{w})$ subtrees contributes $O(w^2)$ crossings from this drawing (the maximum number of crossings for a graph on $O(\sqrt{w})$ vertices drawn with straight-line edges on a circle). Projecting the circle onto a line gives a linear arrangement of the subgraph, and of its planarization, with edge separation number $O(w)$, so by \autoref{lem:cutwidth2carving} the carving width of this subgraph and its planarization is also $O(w)$.

Let $T'$ be the binary tree resulting from $T$ by contracting each $T_i$ into a point. As in \autoref{thm:carving} we draw $T'$ in the plane, replacing each of its vertices by a disk and replacing each of its edges by a rectangle. We place the drawing of the subgraph associated with each subtree $T_i$ into the corresponding disk of $T'$. We replace the (up to two) two dummy vertices representing connections from $T_i$  to other subtrees with a bundle of edges passing from the disk to an adjacent rectangle. As in \autoref{thm:carving} we route edges across each rectangle by monotone curves that cross each other at most once. There are $O(n/\sqrt{w})$ rectangles, each having $O(w^2)$ crossings and carving width at most $w$, so the contributions to the total number of crossings and the carving width from this part of the construction are also $O(w^{3/2}n)$ and $O(w)$ respectively.
\fi

\end{document}